\documentclass[11pt]{article}
\usepackage{geometry}                
\usepackage{graphicx}
\usepackage{amssymb}
\usepackage{amsmath}
\usepackage{epstopdf}
\usepackage[title]{appendix}
\usepackage{amsthm}

\newtheorem{proposition}{Proposition}

\newtheorem{conjecture}{Conjecture}
\newtheorem{lemma}{Lemma}[]
\newtheorem{remark}{Remark}[]

\usepackage{multirow}
\usepackage{setspace}
\usepackage{authblk}

\usepackage[plain]{algorithm2e}
\usepackage{tikz}

\usepackage[all,cmtip]{xy}
\usepackage{lineno}
\title{Inference on autoregulation in gene expression}
\author[1,2,*]{Yue Wang}
\author[3]{Siqi He}
\affil[1]{Department of Computational Medicine, University of California, Los Angeles, California, United States of America}
\affil[2]{Institut des Hautes \'Etudes Scientifiques (IH\'ES), Bures-sur-Yvette, Essonne, France}
\affil[3]{Simons Center for Geometry and Physics, Stony Brook University, Stony Brook, New York, United States of America}
\affil[*]{E-mail address: yuew@g.ucla.edu (Y. W.). ORCID: 0000-0001-5918-7525}
\date{}                                           
\begin{document}
\maketitle

\begin{abstract}
Some genes can promote or repress their own expressions, which is called autoregulation. Although gene regulation is a central topic in biology, autoregulation is much less studied. In general, it is extremely difficult to determine the existence of autoregulation with direct biochemical approaches. Nevertheless, some papers have observed that certain types of autoregulations are linked to noise levels in gene expression. We generalize these results by two propositions on discrete-state continuous-time Markov chains. These two propositions form a simple but robust method to infer the existence of autoregulation from gene expression data. This method only needs to compare the mean and variance of the gene expression level. Compared to other methods for inferring autoregulation, our method only requires non-interventional one-time data, and does not need to estimate parameters. Besides, our method has few restrictions on the model. We apply this method to four groups of experimental data and find some genes that might have autoregulation. Some inferred autoregulations have been verified by experiments or other theoretical works.
	
\end{abstract}

\smallskip
\noindent \textbf{Keywords.} 

\noindent inference; gene expression; autoregulation; Markov chain.

\

\noindent \textbf{Frequently used abbreviations:} 

\noindent GRN: gene regulatory network.

\noindent VMR: variance-to-mean ratio

\section{Introduction}
\label{intro}

In general, genes are transcribed to mRNAs and then translated to proteins. We can use the abundance of mRNA or protein to represent the expression levels of genes. Both the synthesis and degradation of mRNAs/proteins are affected (activated or inhibited) by the expression levels of other genes \cite{karamyshev2018lost}, which is called (mutual) gene regulation. Genes and their regulatory relations form a gene regulatory network (GRN) \cite{cunningham2015mechanisms}, generally represented as a directed graph: each vertex is a gene, and each directed edge is a regulatory relationship. See Fig.~\ref{grn} for an example of GRN.

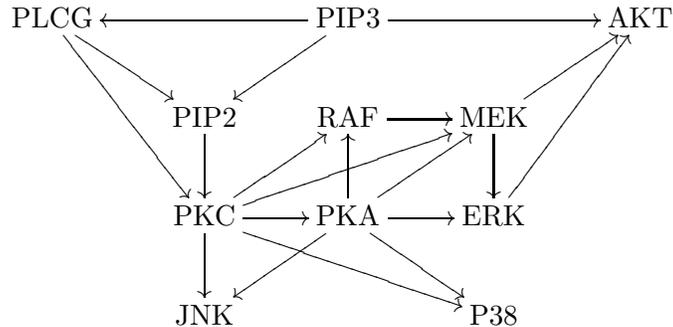
\begin{figure}
	\center
	$\xymatrix{
		\text{PLCG}\ar[rd]\ar[rdd]&&\text{PIP3}\ar[ll]\ar[ld]\ar[rr]&&\text{AKT}\\
		&\text{PIP2}\ar[d]&\text{RAF}\ar[r]&\text{MEK}\ar[ru]\ar[d]&\\
		&\text{PKC}\ar[ru]\ar[rru]\ar[r]\ar[d]\ar[rrd]&\text{PKA}\ar[u]\ar[ru]\ar[r]\ar[rd]\ar[ld]&\text{ERK}\ar[ruu]&\\
		&\text{JNK}&&\text{P38}&
	}$\\
	\caption{An example of GRN in human T cells \cite{werhli2006comparative}. Each vertex is a gene. Each arrow is a regulatory relationship. Notice that it has no directed cycle.}
	\label{grn}
\end{figure}

The expression of one gene could promote/repress its own expression, which is called positive/negative autoregulation \cite{carrier1999investigating}. Autoregulation is very common in \emph{E. coli} \cite{shen2002network}. Positive autoregulation is also called autocatalysis or autoactivation, and negative autoregulation is also called autorepression \cite{baumdick2018conformational,fang2017sirt7}. For instance, HOX proteins form and maintain spatially inhomogeneous expression of HOX genes \cite{sheth2014self}. For genes with position-specific expressions during development, it is common that the increase of one gene can further increase or decrease its level \cite{wang2020biological}.

While countless works infer the regulatory relationships between different genes (GRN structure) \cite{wang2022inference}, determining the existence of autoregulation is an equally important yet less-studied field. Due to technical limitations, it is difficult and sometimes impossible to directly detect autoregulation in experiments. Instead, we can measure gene expression profiles and infer the existence of autoregulation. In this paper, we consider a specific data type: measure the expression levels of certain genes without intervention for a single cell (which reaches stationary) at a single time point, and repeat for many different cells to obtain a probability distribution for expression levels. Such single-cell non-interventional one-time gene expression data can be obtained with a relatively lower cost \cite{luecken2019current}.

With such single-cell level data for one gene $V$, we can calculate the ratio of variance and mean of the expression level (mRNA or protein count). This quantity is called the variance-to-mean ratio (VMR) or the Fano factor. Many papers that study gene expression systems with autoregulations have found that negative autoregulation can decrease noise (smaller VMR), and positive autoregulation can increase noise (larger VMR) \cite{thattai2001intrinsic,swain2004efficient,hornos2005self,munsky2012using,gronlund2013transcription,dessalles2017stochastic,czuppon2018limits}. This means VMR can be used to infer the existence of autoregulation. 

We generalize the above observation and develop two mathematical results that use VMR to determine the existence of autoregulation. They apply to some genes that have autoregulation. For genes without autoregulation, these results cannot determine that autoregulation does not exist. We apply these results to four experimental gene expression data sets and detect some genes that might have autoregulation. 

We start with some setup and introduce our main results (Section~\ref{setup}). Then we cite some previous works on this topic and compare them with our results (Section~\ref{related}). For a single gene that is not regulated by other genes (Section~\ref{auto}) and multiple genes that regulate each other (Section~\ref{multi}), we develop mathematical results to identify the existence of autoregulation. These two mathematical sections can be skipped. We summarize the procedure of our method and apply it to experimental data (Section~\ref{app}). We finish with some conclusions and discussions (Section~\ref{con}).

\section{Setup and main results}
\label{setup}
One possible mechanism of ``the increase of one gene's expression level further increases its expression level'' is a positive feedback loop between two genes \cite{hui2020increased}. Here $V_1$ and $V_2$ promote each other, so that the increase of $V_1$ increases $V_2$, which in return further increases $V_1$. We should not regard this feedback loop as autoregulation. When we define autoregulation for a gene $V$, we should fix environmental factors and other genes that regulate $V$, and observe whether the expression level of $V$ can affect itself. If $V$ is in a feedback loop that contains other genes, then those genes (which regulate $V$ and are regulated by $V$) cannot be fixed when we change $V$. Therefore, it is essentially difficult to determine whether $V$ has autoregulation in this scenario. In the following, we need to assume that $V$ is not contained in a feedback loop that involves other genes.

The actual gene expression mechanism might be complicated. Besides other genes/factors that can regulate a gene, for a gene $V$ itself, it might switch between inactivated (off) and activated (on) states \cite{cao2020analytical}. These states correspond to different transcription rates to produce mRNAs. When mRNAs are translated into proteins, those proteins might affect the transition of gene activation states, which forms autoregulation \cite{firman2018maximum}. See Fig.~\ref{mech} for an illustration. Therefore, for a gene $V$, we should regard the gene activation state, mRNA count, and protein count as a triplet of random variables $G,M,P$, which depend on each other. 

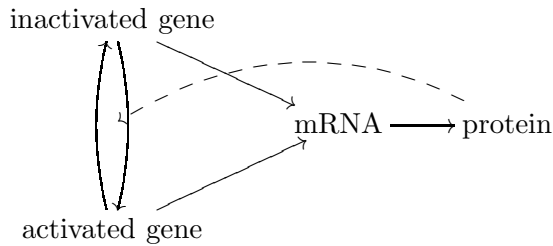
\begin{figure}
	\center
	$\xymatrix{
		\text{inactivated gene}\ar[rd]\ar@/^/[dd]&&\\
		&\text{mRNA}\ar[r]&\text{protein}\ar@{-->}@/_2pc/[ll]\\
		\text{activated gene}\ar[ru]\ar@/^/[uu]&
	}$\\
	\caption{The mechanism of gene expression. A gene might switch between inactivated state and activated state. Gene is transcribed into mRNAs, which are translated into proteins. Proteins might (auto)regulate the state transition of the corresponding gene.}
	\label{mech}
\end{figure}

When we fix environmental factors and other genes that affect $V$, the triplet $G,M,P$ should follow a continuous-time Markov chain. The state space is on/off (for $G$) or the mRNA/protein count on $\mathbb{Z}$ (for $M,P$). When we consider the expression level $M$ or $P$ (but do not control $G$), sometimes itself still follows a Markov chain, and we call this scenario ``\textbf{autonomous}''. In other cases, $M$ or $P$ itself is no longer Markovian, and we call this scenario ``\textbf{non-autonomous}''. We need to consider the triplet $G,M,P$ in the non-autonomous scenario.

For the autonomous scenario, we can fully classify autoregulation for a gene $V$. Assume environmental factors and other genes that affect the expression of $V$ are kept at constants. Define the expression level (mRNA count for example) of one cell to be $X=n$, the mRNA synthesis rate at $X=n-1$ to be $f_n$, and the degradation rate for each mRNA molecule at $X=n$ to be $g_n$. This is a standard continuous-time Markov chain on $\mathbb{Z}$ with transition rates 
\[\frac{1}{\Delta t}\mathbb{P}[X(t+\Delta t)=n\mid X(t)=n-1]=f_n,\]
\[\frac{1}{\Delta t}\mathbb{P}[X(t+\Delta t)=n-1\mid X(t)=n]=ng_n.\]
Define the relative growth rate $h_n=f_n/g_n$. If there is \textbf{no autoregulation}, then $h_n$ is a constant. \textbf{Positive autoregulation} means $h_n>h_{n-1}$ for some $n$, so that $f_n>f_{n-1}$ and/or $g_n<g_{n-1}$; \textbf{negative autoregulation} means $h_n<h_{n-1}$ for some $n$, so that $f_n<f_{n-1}$ and/or $g_n>g_{n-1}$. Notice that we can have $h_n>h_{n-1}$ for some $n$ and $h_{n'}<h_{n'-1}$ for some other $n'$, meaning that positive autoregulation and negative autoregulation can both exist for the same gene, but occur at different expression levels.

For the non-autonomous scenario, we can still define autoregulation. Consider the expression level $X$ of $V$ (mRNA count or protein count) and its interior factor $I$. If $X$ is the mRNA count, then $I$ is the gene state; if $X$ is the protein count, then $I$ is the gene state and the mRNA count. If there is \textbf{no autoregulation}, then $X$ cannot affect $I$, and for each value of $I$, the relative growth rate $h_n$ of $X$ is a constant. If $X$ can affect $I$, or $h_n$ is not a constant, then there is \textbf{autoregulation}. When $X$ can affect $I$, it is not always easy to distinguish between positive autoregulation and negative autoregulation.

Quantitatively, for the autonomous scenario, when we fix other factors that might regulate this gene $V$, if $V$ has no autoregulation, then $h_n=f_n/g_n$ is a constant $h$ for all $n$. In this case, the stationary distribution of $V$ satisfies $\mathbb{P}(X=n)/\mathbb{P}(X=n-1)=h/n$, meaning that the distribution is Poisson with parameter $h$, $\mathbb{P}(X=n)=h^ne^{-h}/n!$, and $\text{VMR}=1$. If there exists positive autoregulation of certain forms, $\text{VMR}>1$; if there exists negative autoregulation of certain forms, $\text{VMR}<1$. However, such results are derived by assuming that $f_n,g_n$ take certain functional forms, such as linear functions \cite{paulsson2005models,ramos2015gene}, quadratic functions \cite{giovanini2020comparative}, or Hill functions \cite{stewart2013under}. There are other papers that consider Markov chain models in gene expression/regulation \cite{jia2017simplification,sharma2014markov,shmulevich2003steady,chen2020limit,shen2019distributed,ko2019markov}, but the role of VMR is not thoroughly studied.

In this paper, we generalize the above result of inferring autoregulation with VMR by dropping the restrictions on parameters. Consider a gene $V$ in a known GRN, and assume it is not regulated by other genes, or assume other factors that regulate $V$ are fixed. Assume we have the \textbf{autonomous} scenario, meaning that its expression level $X=n$ satisfies a general Markov chain with synthesis rate $f_n$ and per molecule degradation rate $g_n$. We do not add any restrictions on $f_n$ and $g_n$. Use the single-cell non-interventional one-time gene expression data to calculate the VMR of $V$. Proposition~\ref{prop2} states that $\text{VMR}>1$ or $\text{VMR}<1$ means the existence of positive/negative autoregulation.

Nevertheless, the autonomous condition requires some assumptions, and often does not hold in reality \cite{bokes2012multiscale,jia2017emergent,jia2020kinetic,jia2017simplification}. Consider a gene $V$ that is not regulated by other genes, and has no autoregulation. The mRNA count or the protein count is regulated by the gene activation state, which cannot be fixed. Due to this non-controllable factor, there might be transcriptional bursting \cite{shahrezaei2008analytical,dobrinic2021prc1} or translational bursting \cite{cagnetta2019noncanonical}, where transcription or translation can occur in bursts, and we have $\text{VMR}>1$. This does not mean that Proposition~\ref{prop2} is wrong. Instead, it means that the expression level itself is not Markovian, and the scenario is non-autonomous. In this scenario, we should apply Proposition~\ref{np}, described below, which states that no autoregulation means $\text{VMR}\ge 1$.

We extend the idea of inferring autoregulation with VMR to a gene that is regulated by other genes, or with non-autonomous expression. Consider a gene $V'$ in a known GRN. Assume $V'$ is not contained in a feedback loop, and assume $g_n$, the per molecule degradation rate of $V'$, is not regulated by other genes or its interior factors. We do not add any restrictions on the synthesis rate $f_n$. Proposition~\ref{np} states that if $V'$ has no autoregulation, then $\text{VMR}\ge 1$. Therefore, $\text{VMR}<1$ means autoregulation for $V'$. The conclusion ``$\text{VMR}<1$ means autoregulation'' has been observed by Munsky et al. \cite{munsky2012using} for a single gene that is non-autonomous.

In the scenario that Proposition~\ref{np} may apply, if $\text{VMR}\ge 1$, Proposition~\ref{np} cannot determine whether autoregulation exists. In fact, with VMR, or even the full probability distribution, we might not distinguish a non-autonomous system with autoregulation from a non-autonomous system without autoregulation, which both have $\text{VMR}\ge 1$ \cite{cao2018linear}. In the non-autonomous scenario, we only focus on the less complicated case of $\text{VMR}<1$, and derive Proposition~\ref{np} that firmly links VMR and autoregulation.

In reality, Proposition~\ref{prop2} and Proposition~\ref{np} can only apply to a few genes (which are not regulated by other genes or have $\text{VMR}<1$), and they cannot determine negative results. Thus the inference results about autoregulation are a few ``yes'' and many ``we do not know''. Besides, for the results inferred by Proposition~\ref{prop2}, especially those with $\text{VMR}>1$ (positive autoregulation), we cannot verify whether their expression is autonomous, and the inference results are less reliable. 

Current experimental methods can hardly determine the existence of autoregulation, and to determine that a gene does not have autoregulation is even more difficult. Therefore, about whether genes in a GRN have autoregulation, experimentally, we do not have ``yes'' or ``no'', but a few ``yes'' and many ``we do not know''. Thus there is no gold standard to thoroughly evaluate the performance of our inference results. We can only report that some genes inferred by our method to have autoregulation are also verified by experiments or other inference methods to have autoregulation.

\section{Related works}
\label{related}
There are other mathematical approaches to infer the existence of autoregulation in gene expression \cite{sanchez2018bayesian,xing2005causal,feigelman2016analysis,veerman2021parameter,jia2018relaxation,zhou2012analytical,jia2020small,jia2020dynamical}. We introduce some works and compare them with our method. (\textbf{A}) Sanchez-Castillo et al. \cite{sanchez2018bayesian} considered an autoregressive model for multiple genes. This method (1) needs time series data; (2) requires the dynamics to be linear; (3) estimates a group of parameters. (\textbf{B}) Xing et al. \cite{xing2005causal} applied causal inference to a complicated gene expression model. This method (1) needs promoter sequences and information on transcription factor binding sites; (2) requires linearity for certain steps; (3) estimates a group of parameters. (\textbf{C}) Feigelman et al. \cite{feigelman2016analysis} applied a Bayesian method for model selection. This method (1) needs time series data; (2) estimates a group of parameters. (\textbf{D}) Veerman et al. \cite{veerman2021parameter} considered the probability-generating function of a propagator model. This method (1) needs time series data; (2) estimates a group of parameters; (3) needs to approximate a Cauchy integral. (\textbf{E}) Jia et al. \cite{jia2018relaxation} compared the relaxation rate with degradation rate. This method (1) needs interventional data; (2) only works for a single gene that is not regulated by other genes; (3) requires that the per molecule degradation rate is a constant.

Compared to other methods, our method has some advantages: (1) Our method uses non-interventional one-time data. Time series data require measuring the same cell multiple times without killing it, and interventional data require some techniques to interfere with gene expression, such as gene knockdown. Therefore, non-interventional one-time data used in our method are much easier and cheaper to obtain. (2) Our method does not estimate parameters, and only calculate the mean and variance of the expression level. Some other methods need to estimate many parameters or approximate some complicated quantities, meaning that they need large data size and high data accuracy. Therefore, our method is easy to calculate, and need lower data accuracy and smaller data size. (3) Our method has few restrictions on the model, making them applicable to various scenarios with different dynamics. In sum, our method is simple and universal, and have lower requirements on data quality.

Compared to other methods, our method has some disadvantages: (1) The GRN structure needs to be known. (2) Our method does not work for certain genes, depending on regulatory relationships. Proposition~\ref{prop2} only works for a gene that is not regulated by other genes, and we require its expression to be autonomous; Proposition~\ref{np} only works for a gene that is not in a feedback loop. (3) Proposition~\ref{np} requires the per molecule degradation rate to be a constant, and it cannot provide information about autoregulation if $\text{VMR}\ge 1$. (4) Our method only works for cells at equilibrium. Thus time series data that contain time-specific information cannot be utilized other than treated as one-time data. With just stationary distribution, sometimes it is impossible to build the causal relationship (including autoregulation) \cite{wang2020causal}. Thus with this data type, some disadvantages are inevitable.

\section{Scenario of a single isolated gene}
\label{auto}

\subsection{Setup}
We first consider the expression level (e.g., mRNA count) of one gene $V$ in a single cell. At the single-cell level, gene expression is essentially stochastic, and we use a random variable $X$ to represent the mRNA count of $V$. We assume $V$ is not in a feedback loop. We also assume all environmental factors and other genes that can affect $X$ are kept at constant levels, so that we can focus on $V$ alone. This can be achieved if no other genes point to gene $V$ in the GRN, such as PIP3 in Fig.~\ref{grn}. Then we assume that the expression of $V$ is autonomous, thus $X$ satisfies a time-homogeneous Markov chain defined on $\mathbb{Z}^*$. 

Assume that the mRNA synthesis rate at $X(t)=n-1$, namely the transition rate from $X=n-1$ to $X=n$, is $f_n>0$. Assume that with $n$ mRNA molecules, the degradation rate for each mRNA molecule is $g_n>0$. Then the overall degradation rate at $X(t)=n$, namely the transition rate from $X=n$ to $X=n-1$, is $g_nn$. The associated master equation is 
\begin{equation}
	\begin{split}
		\mathrm{d}\mathbb{P}[X(t)=n]/\mathrm{d}t=&\mathbb{P}[X(t)=n+1]g_{n+1}(n+1)+\mathbb{P}[X(t)=n-1]f_n\\
		&-\mathbb{P}[X(t)=n](f_{n+1}+g_nn).
	\end{split}
	\label{eq0}
\end{equation}
Define the relative growth rate $h_n=f_n/g_n$. We assume that $h_n$ has a finite upper bound. This means that as time tends to infinity, the process reaches equilibrium. Thus at equilibrium, (1) the stationary probability distribution $P_n=\lim_{t\to\infty}\mathbb{P}[X(t)=n]$ exists, and $P_n=P_{n-1}h_n/n$; (2) the mean $\mathbb{E}(X)$ and the variance $\sigma^2(X)$ are finite \cite{wang2022discrete}. 

If $h_n>h_{n-1}$ for some $n$, then there exists positive autoregulation. If $h_n<h_{n-1}$ for some $n$, then there exists negative autoregulation. If there is no autoregulation, then  $h_n$ is a constant $h$, and the stationary distribution is Poisson with parameter $h$. In this setting, positive autoregulation and negative autoregulation might coexist, meaning that $h_{n+1}<h_n$ for some $n$ and $h_{n'+1}>h_{n'}$ for some $n'$.

\subsection{Theoretical results}

With single-cell non-interventional one-time gene expression data for one gene, we have the stationary distribution of the Markov chain $X$. We can infer the existence of autoregulation with the VMR of $X$, defined as $\text{VMR}(X)=\sigma^2(X)/\mathbb{E}(X)$. The idea is that if we let $f_n$ increase/decrease with $n$, and control $g_n$ to make $\mathbb{E}(X)$ invariant, then the variance $\sigma^2(X)$ increases/decreases \cite{wang2018some}. We shall prove that $\text{VMR}>1$ means positive autoregulation, and $\text{VMR}<1$ means negative autoregulation. Notice that $\text{VMR}>1$ does not exclude the possibility that negative autoregulation exists for some expression level. This also applies to $\text{VMR}<1$ and positive autoregulation.

We can illustrate this result with a linear model: set $f_n=k+b(n-1)$, $g_n=c$. Here $b$ (can be positive or negative) is the strength of autoregulation, and $c$ satisfies $c>0$ and $c-b>0$. Multiply Eq.~\ref{eq0} by $n$ and $n(n-1)$ and take summation, then we can calculate that $\text{VMR}=1+b/(c-b)$. Therefore, $\text{VMR}>1$ means positive autoregulation, $b>0$; $\text{VMR}<1$ means negative autoregulation, $b<0$; $\text{VMR}=1$ means no autoregulation, $b=0$. 

\begin{lemma}
	Consider the Markov chain model for one gene with general transition coefficients $f_n,g_n$, described by Eq.~\ref{eq0}. Calculate $\text{VMR}(X)$ at stationary. (1) Assume $h_{n+1}\ge h_n$ for all $n$. We have $\text{VMR}(X)\ge 1$; moreover, $\text{VMR}(X)= 1$ if and only if $h_{n+1}= h_n$ for all $n$. (2) Assume $h_{n+1}\le h_n$ for all $n$. We have $\text{VMR}(X)\le 1$; moreover, $\text{VMR}(X)= 1$ if and only if $h_{n+1}= h_n$ for all $n$. 
	\label{lemma0}
\end{lemma}
From Lemma~\ref{lemma0}, we can directly obtain the following proposition.
\begin{proposition}
	\label{prop2}	
	In the setting of Lemma~\ref{lemma0}, (1) If $\text{VMR}(X)>1$, then there exist values of $n$ for which $h_{n+1}>h_n$; thus this gene has positive autoregulation. (2) If $\text{VMR}(X)<1$, then there exist values of $n$ for which $h_{n+1}<h_n$; thus this gene has negative autoregulation. (3) If $\text{VMR}(X)=1$, then either (A) $h_{n+1}=h_n$ for all $n$, meaning that this gene has no autoregulation; or (B) $h_{n+1}<h_n$ for some $n$ and $h_{n'+1}>h_{n'}$ for some $n'$, meaning that this gene has both positive and negative autoregulation (at different expression levels).
\end{proposition}
\begin{remark}
Proposition~\ref{prop2} requires that the gene expression is autonomous. In reality, many genes are non-autonomous, and transcriptional/translational bursting can make the VMR to be larger than $100$ \cite{paulsson2005models}.
\end{remark}
\begin{remark}
Results similar to Proposition~\ref{prop2} have been proven in a non-autonomous model of gene expression \cite{jia2017stochastic}.
\end{remark}
\begin{proof}[Proof of Lemma~\ref{lemma0}]
	Define $\lambda=-\log P_0$, so that $P_0=\exp(-\lambda)$. Define $d_n=\prod_{i=1}^{n}h_i>0$ and stipulate that $d_0=1$. We can see that 
	\[\frac{d_nd_{n+2}}{d_{n+1}^2}=\frac{h_{n+2}}{h_{n+1}}.\]
	Also, 
	\[P_n=P_{n-1}f_n/(g_n n)=P_{n-1}h_n/n=\cdots=P_0(\prod_{i=1}^{n}h_i)/n!\ =e^{-\lambda}\frac{d_n}{n!}.\]
	Then 
	\begin{equation*}
		\begin{split}
			\mathbb{E}(X^2)-\mathbb{E}(X)&=\sum_{n=1}^{\infty}(n^2-n)P_n=e^{-\lambda}\sum_{n=1}^{\infty}(n^2-n)\frac{d_n}{n!}\\
			&=e^{-\lambda}\sum_{n=2}^{\infty}\frac{d_n}{(n-2)!}=e^{-\lambda}\sum_{n=0}^{\infty}\frac{d_{n+2}}{n!},
		\end{split}
	\end{equation*}
	\[[\mathbb{E}(X)]^2=\left(\sum_{n=1}^{\infty}nP_n\right)^2=e^{-2\lambda}\left(\sum_{n=1}^{\infty}n\frac{d_n}{n!}\right)^2=e^{-2\lambda}\left(\sum_{n=0}^{\infty}\frac{d_{n+1}}{n!}\right)^2.\]
	Besides,
	\[1=\sum_{n=0}^{\infty}P_n=e^{-\lambda}\sum_{n=0}^{\infty}\frac{d_n}{n!}.\]
	Now we have 
	\[\mathbb{E}(X^2)-\mathbb{E}(X)-[\mathbb{E}(X)]^2=e^{-2\lambda}\left(\sum_{n=0}^{\infty}\frac{d_n}{n!}\right)\left(\sum_{n=0}^{\infty}\frac{d_{n+2}}{n!}\right)-e^{-2\lambda}\left(\sum_{n=0}^{\infty}\frac{d_{n+1}}{n!}\right)^2.\]
	
	(1) Assume $h_{n+1}\ge h_n$ for all $n$. Then 
	
	\begin{equation}
		\label{eqe}
		\begin{split}
			&\mathbb{E}(X^2)-\mathbb{E}(X)-[\mathbb{E}(X)]^2\\
			\ge & e^{-2\lambda}\left(\sum_{n=0}^{\infty}\frac{\sqrt{d_nd_{n+2}}}{n!}\right)^2- e^{-2\lambda}\left(\sum_{n=0}^{\infty}\frac{d_{n+1}}{n!}\right)^2\ge 0.
		\end{split}
	\end{equation}
	Here the first inequality is from the Cauchy inequality, and the second inequality is from $d_nd_{n+2}\ge d_{n+1}^2$ for all $n$. Then $\text{VMR}(X)=\{\mathbb{E}(X^2)-[\mathbb{E}(X)]^2\}/\mathbb{E}(X)\ge 1$. Equality holds if and only if $d_n/d_{n+2}=d_{n+1}/d_{n+3}$ for all $n$ (the first inequality of Eq.~\ref{eqe}) and $d_nd_{n+2}=d_{n+1}^2$ for all $n$ (the second inequality of Eq.~\ref{eqe}). The equality condition is equivalent to $h_{n+1}=h_n$ for all $n$.
	
	(2) Assume $h_{n+1}\le h_n$ for all $n$. Then $d_{n+2}/d_{n+1}\le d_{n+1}/d_n$, and $d_n\le h_1^n$ for all $n$. Define 
	\[H(t)=\sum_{n=0}^{\infty}\frac{d_n}{n!}t^n.\]
	Since $0<d_n\le h_1^n$, this series converges for all $t\in\mathbb{C}$, so that $H(t)$ is a well-defined analytical function on $\mathbb{C}$, and 
	\[H'(t)=\sum_{n=0}^{\infty}\frac{d_{n+1}}{n!}t^n,\ \text{ and }\ H''(t)=\sum_{n=0}^{\infty}\frac{d_{n+2}}{n!}t^n.\]
	In the following, we only consider $H(t),H'(t),H''(t)$ as real functions for $t\in\mathbb{R}$.
	
	To prove $\text{VMR}(X)\le 1$, we just need to prove $\mathbb{E}(X^2)-\mathbb{E}(X)-[\mathbb{E}(X)]^2=e^{-2\lambda}\{H(1)H''(1)-[H'(1)]^2\}\le 0$. However, we shall prove $H''(t)H(t)\le [H'(t)]^2$ for all $t\in \mathfrak{I}$, where $\mathfrak{I}=(a,b)$ is a fixed interval in $\mathbb{R}$ with $0<a<1$ and $1<b<\infty$. Thus $t=1$ is an interior point of $\mathfrak{I}$. Since $H(t),H'(t),H''(t)$ have positive lower bounds on $\mathfrak{I}$, the following statements are obviously equivalent: (i) $H''(t)H(t)\le [H'(t)]^2$ for all $t\in\mathfrak{I}$; (ii) $\{\log[H'(t)/H(t)]\}'\le 0$ for all $t\in\mathfrak{I}$; (iii) $\log[H'(t)/H(t)]$ is non-increasing on $\mathfrak{I}$; (iv) $H'(t)/H(t)$ is non-increasing on $\mathfrak{I}$. To prove (i), we just need to prove (iv).
	
	Consider any $t_1,t_2\in \mathfrak{I}$ with $t_1\le t_2$ and any $p,q\in \mathbb{N}$ with $p\ge q$. Since $d_{p+1}/d_p\le d_{q+1}/d_q$, and $t_1^{p-q}\le t_2^{p-q}$, we have 
	\[
	d_pd_qt_1^qt_2^q(\frac{d_{p+1}}{d_p}-\frac{d_{q+1}}{d_q})(t_1^{p-q}-t_2^{p-q})\ge 0,
	\]
	which means 
	\[
	d_{p+1}d_qt_1^pt_2^q+d_{q+1}d_p t_1^qt_2^p\geq d_{p+1}d_qt_2^pt_1^q+d_{q+1}d_pt_2^qt_1^p.
	\]
	Sum over all $p,q\in \mathbb{N}$ with $p\ge q$ to obtain
	\begin{equation*}
		\begin{split}
			H'(t_1)H(t_2)&=(\sum_{n=0}^{\infty}\frac{d_{n+1}}{n!}t_1^n)(\sum_{n=0}^{\infty}\frac{d_n}{n!}t_2^n)\\
			&\ge (\sum_{n=0}^{\infty}\frac{d_{n+1}}{n!}t_2^n)(\sum_{n=0}^{\infty}\frac{d_n}{n!}t_1^n)=H'(t_2)H(t_1).
		\end{split}
	\end{equation*}
	Thus $H'(t_1)/H(t_1)\ge H'(t_2)/H(t_2)$ for all $t_1,t_2\in \mathfrak{I}$ with $t_1\le t_2$. This means $H''(t)H(t)\le [H'(t)]^2$ for all $t\in\mathfrak{I}$, and $\text{VMR}(X)\le 1$.
	
	About the condition for the equality to hold, assume $h_{n'+1}<h_{n'}$ for a given $n'$. Then 
	\[
	d_{n'}d_{n'-1}t_1^{n'-1}t_2^{n'-1}(\frac{d_{n'+1}}{d_{n'}}-\frac{d_{n'}}{d_{n'-1}})(t_1-t_2)\ge C(t_2-t_1)
	\]
	for all $t_1,t_2\in \mathfrak{I}$ with $t_1\le t_2$ and a constant $C$ that does not depend on $t_1,t_2$. Therefore, 
	\begin{equation*}
		\begin{split}
			&[H'(t_1)/H(t_1)-H'(t_2)/H(t_2)]\cdot[H(t_1)H(t_2)]\\
			=&(\sum_{n=0}^{\infty}\frac{d_{n+1}}{n!}t_1^n)(\sum_{n=0}^{\infty}\frac{d_n}{n!}t_2^n)- (\sum_{n=0}^{\infty}\frac{d_{n+1}}{n!}t_2^n)(\sum_{n=0}^{\infty}\frac{d_n}{n!}t_1^n)\\
			\ge &d_{n'}d_{n'-1}t_1^{n'-1}t_2^{n'-1}(\frac{d_{n'+1}}{d_{n'}}-\frac{d_{n'}}{d_{n'-1}})(t_1-t_2)\\
			\ge& C(t_2-t_1).
		\end{split}
	\end{equation*}
	Since $H(t)$ has a finite positive upper bound $A$ and a positive lower bound $B$ on $\mathfrak{I}$, we have \[H'(t_1)/H(t_1)-H'(t_2)/H(t_2)\ge C(t_2-t_1)/A^2,\] 
	meaning that 
	\[\forall t\in\mathfrak{I},\,\,[H'(t)/H(t)]'=\{H(t)H''(t)-[H'(t)]^2\}/[H(t)]^2\le -C/A^2,\]
	and thus 
	\[\forall t\in\mathfrak{I},\,\,H(t)H''(t)-[H'(t)]^2\le -CB^2/A^2<0.\]
	Therefore, $\mathbb{E}(X^2)-\mathbb{E}(X)-[\mathbb{E}(X)]^2=e^{-2\lambda}\{H(1)H''(1)-[H'(1)]^2\}<0$, and $\text{VMR}(X)<1$. 
	
	We have proved in (1) that if $h_{n+1}=h_n$ for all $n$, then $\text{VMR}(X)=1$. Thus when $h_{n+1}\le h_n$ for all $n$, $\text{VMR}(X)=1$ if and only if $h_{n+1}=h_n$ for all $n$.
\end{proof}

In sum, for the Markov chain model of one gene (by assuming the expression to be autonomous), when we have the stationary distribution from single-cell non-interventional one-time gene expression data, we can calculate the VMR of $X$. $\text{VMR}(X)>1$ means the existence of positive autoregulation, and $\text{VMR}(X)<1$ means the existence of negative autoregulation. $\text{VMR}(X)=1$ means either (1) no autoregulation exists; or (2) both positive autoregulation and negative autoregulation exist (at different expression levels).

\section{Scenario of multiple entangled genes}
\label{multi}
\subsection{Setup}
We consider $m$ genes $V_1,\ldots,V_m$ for a single cell. Denote their expression levels by random variables $X_1,\ldots,X_m$. The change of $X_i$ can depend on $X_j$ (mutual regulation) and $X_i$ itself (autoregulation). Since these genes regulate each other, and their expression levels are not fixed, we cannot consider them separately. If the expression of gene $V_k$ is non-autonomous, we also need to add its interior factors (gene activation state, etc.) into $X_1,\ldots,X_m$.

We can use a continuous-time Markov chain on $(\mathbb{Z}^*)^m$ to describe the dynamics. Each state of this Markov chain, $(X_1=n_1,\ldots,X_i=n_i,\ldots,X_m=n_m)$, can be abbreviated as $\boldsymbol{n}=(n_1,\ldots,n_i,\ldots,n_m)$. For gene $V_i$, the transition rate of $n_i-1\to n_i$ is $f_i(\boldsymbol{n})$, and the the transition rate of $n_i\to n_i-1$ is $g_i(\boldsymbol{n})n_i$. The master equation of this process is 
\begin{equation}
	\label{me}
	\begin{split}
		\mathrm{d}\mathbb{P}(\boldsymbol{n})/\mathrm{d}t=&\sum_i\mathbb{P}(n_1,\ldots,n_i+1,\ldots,n_m)g_i(n_1,\ldots,n_i+1,\ldots,n_m)(n_i+1)\\
		&+\sum_i\mathbb{P}(n_1,\ldots,n_i-1,\ldots,n_m)f_i(\boldsymbol{n})\\
		&-\mathbb{P}(\boldsymbol{n})\sum_i[f_i(n_1,\ldots,n_i+1,\ldots,n_m)+g_i(\boldsymbol{n})n_i].
	\end{split}		
\end{equation}	
Define $\boldsymbol{n}_{\bar{i}}=(n_1,\ldots,n_{i-1},n_{i+1},\ldots,n_m)$. Define $h_i(\boldsymbol{n})=f_i(\boldsymbol{n})/g_i(\boldsymbol{n})$ to be the relative growth rate of gene $V_i$. Autoregulation means for some fixed $\boldsymbol{n}_{\bar{i}}$,  $h_i(\boldsymbol{n})$ is (locally) increasing/decreasing  with $n_i$, thus $f_i(\boldsymbol{n})$ increases/decreases and/or $g_i(\boldsymbol{n})$ decreases/increases with $n_i$. For the non-autonomous scenario, another possibility for autoregulation is that $V_i$ can affect its interior factors.

\subsection{Theoretical results}
With expression data for multiple genes, there are various methods to infer the regulatory relationships between different genes, so that the GRN can be reconstructed \cite{wang2022inference}. In the GRN, if there is a directed path from gene $V_i$ to gene $V_j$, meaning that $V_i$ can directly or indirectly regulate $V_j$, then $V_i$ is an ancestor of $V_j$, and $V_j$ is a descendant of $V_i$.

Fix a gene $V_k$ in a GRN. We first consider a simple case that $V_k$ is not contained in any directed cycle (feedback loop), which means no gene is both an ancestor and a descendant of $V_k$, such as PIP2 in Fig.~\ref{grn}. This means $V_k$ itself is a strongly connected component of the GRN. This condition is automatically satisfied if the GRN has no directed cycle. If the expression of $V_k$ is non-autonomous, we need to add the interior factors of $V_k$ into $V_1,\ldots,V_m$, and it is acceptable that $V_k$ regulates its interior factors. In this case, we can prove that if $V_k$ does not regulate itself, meaning that $h_k(\boldsymbol{n})$ is a constant for fixed $\boldsymbol{n}_{\bar{k}}$ and different $n_k$, and $X_k$ does not affect its interior factors (if non-autonomous), then $\text{VMR}(X_k)\ge 1$. The reason is that $\text{VMR}<1$ requires either a feedback loop or autoregulation. We need to assume that the per molecule degradation rate $g_k(\cdot)$ for $V_k$ is not affected by $V_1,\ldots,V_m$, which is not always true in reality \cite{karamyshev2018lost}. With this result, when $\text{VMR}<1$, there might be autoregulation.
\begin{proposition}
	\label{np}
	Consider the Markov chain model for multiple genes, described by Eq.~\ref{me}. Assume the GRN has no directed cycle, or at least there is no directed cycle that contains gene $V_k$. Assume $g_k(\cdot)$ is a constant for all $\boldsymbol{n}$. If $V_k$ has no autoregulation, meaning that $h_k(\cdot)$ and $f_k(\cdot)$ do not depend on $n_k$, and $V_k$ does not regulate its interior factors, then $V_k$ has $\text{VMR}\ge 1$. Therefore, $V_k$ has $\text{VMR}< 1$ means $V_k$ has autoregulation.
\end{proposition}
This Proposition is in an unpublished work by Paulsson et al., who study a similar problem \cite{hilfinger2016constraints,yan2019kinetic}. It also appears in a preprint by Mahajan et al. \cite{mahajan2021topological}, but the proof is based on a linear approximation. We propose a rigorous proof independently.
\begin{proof}
	Denote the expression level of $V_k$ by $W$. Assume the ancestors of $V_k$ are $V_1,\ldots,V_l$. For simplicity, denote the expression levels of $V_1,\ldots,V_l$ by a (high-dimensional) random variable $Y$. Assume $V_k$ has no autoregulation. Since $V_k$ does not regulate $V_1,\ldots,V_l$, $W$ does not affect $Y$. Denote the transition rate from $Y=i$ to $Y=j$ by $q_{ij}\ge 0$. Stipulate that $q_{ii}=-\sum_{j\ne i}q_{ij}$. When $Y=i$, the transition rate from $W=n$ to $W=n+1$ is $F_i$ (does not depend on $n$), and the transition rate from $W=n$ to $W=n-1$ is $G$. 
	
	The master equation of this process is 
	\begin{equation*}
		\begin{split}
			&\mathrm{d}\mathbb{P}[W(t)=n,Y(t)=i]/\mathrm{d}t\\
			=&\mathbb{P}[W(t)=n-1,Y(t)=i]F_i+\mathbb{P}[W(t)=n+1,Y(t)=i]G(n+1)\\
			&+\sum_{j\ne i}\mathbb{P}[W(t)=n,Y(t)=j]q_{ji}-\mathbb{P}[W(t)=n,Y(t)=i](F_i+Gn+\sum_{j\ne i}q_{ij}).
		\end{split}
	\end{equation*}
	Assume there is a unique stationary probability distribution $P_{n,i}=\lim_{t\to \infty}\mathbb{P}[W(t)=n,Y(t)=i]$. Then we have 
	\begin{equation}
		\label{eq1}
		P_{n,i}\Big[F_i+Gn+\sum_{j}q_{ij}\Big]=P_{n-1,i}F_i+P_{n+1,i}G(n+1)+\sum_{j}P_{n,j}q_{ji}.
	\end{equation}
	Define $P_i=\sum_n P_{n,i}$. Sum over $n$ for Eq.~\ref{eq1} to obtain 
	\begin{equation}
		\label{eqz1}
		P_i\sum_{j}q_{ij}=\sum_{j}P_jq_{ji},
	\end{equation}
	meaning that $P_i$ is the stationary probability distribution of $Y$.
	
	Define $W_i$ to be $W$ conditioned on $Y=i$ at stationary. Then $\mathbb{P}(W_i=n)=\mathbb{P}(W=n\mid Y=i)=P_{n,i}/P_i$, and $\mathbb{E}(W_i)=\sum_n n P_{n,i}/P_i$. Multiply Eq.~\ref{eq1} by $n$ and sum over $n$ to obtain 
	\begin{equation}
		\label{eq2}
		\Big(G+\sum_j q_{ij}\Big)P_i\mathbb{E}(W_i)=F_iP_i+\sum_j q_{ji}P_j\mathbb{E}(W_j).
	\end{equation}
	Sum over $i$ for Eq.~\ref{eq2} to obtain
	\begin{equation}
		\label{eq3}
		G\sum_iP_i\mathbb{E}(W_i)=\sum_iF_iP_i.
	\end{equation}
	Multiply Eq.~\ref{eq1} by $n^2$ and sum over $n$ to obtain 
	\begin{equation}
		\label{eq4}
		\Big(2G+\sum_j q_{ij}\Big)P_i\mathbb{E}(W_i^2)=F_iP_i+(2F_i+G)P_i\mathbb{E}(W_i)+\sum_j q_{ji}P_j\mathbb{E}(W_j^2).
	\end{equation}
	Sum over $i$ for Eq.~\ref{eq4} to obtain
	\begin{equation}
		\label{eq5}
		2G\sum_iP_i\mathbb{E}(W_i^2)=\sum_iF_iP_i+2\sum_i F_iP_i\mathbb{E}(W_i)+G\sum_iP_i\mathbb{E}(W_i).
	\end{equation}
	Multiply Eq.~\ref{eq2} by $\mathbb{E}(W_i)$ and sum over $i$ to obtain
	\begin{equation}
		\label{eqz2}
		\begin{split}
			&G\sum_iP_i[\mathbb{E}(W_i)]^2+\sum_{i,j}P_iq_{ij}[\mathbb{E}(W_i)]^2\\
			=&\sum_iF_iP_i\mathbb{E}(W_i)+\sum_{i,j}P_jq_{ji}\mathbb{E}(W_i)\mathbb{E}(W_j).
		\end{split}
	\end{equation}
	Then we have
	\begin{equation} 
		\label{eq6}
		\begin{split}
			&\sum_iF_iP_i\mathbb{E}(W_i)-G\sum_iP_i[\mathbb{E}(W_i)]^2 \\
			=&\sum_{i,j}P_iq_{ij}[\mathbb{E}(W_i)]^2-\sum_{i,j}P_jq_{ji}\mathbb{E}(W_i)\mathbb{E}(W_j)\\
			=&\frac{1}{2}\Big\{\sum_{i,j}P_iq_{ij}[\mathbb{E}(W_i)]^2+\sum_i[\mathbb{E}(W_i)]^2\sum_jP_iq_{ij}-2\sum_{i,j}P_iq_{ij}\mathbb{E}(W_i)\mathbb{E}(W_j)\Big\}\\
			=&\frac{1}{2}\Big\{\sum_{i,j}P_iq_{ij}[\mathbb{E}(W_i)]^2+\sum_i[\mathbb{E}(W_i)]^2\sum_jP_jq_{ji}-2\sum_{i,j}P_iq_{ij}\mathbb{E}(W_i)\mathbb{E}(W_j)\Big\}\\
			=&\frac{1}{2}\Big\{\sum_{i,j}P_iq_{ij}[\mathbb{E}(W_i)]^2+\sum_{i,j}P_iq_{ij}[\mathbb{E}(W_j)]^2-2\sum_{i,j}P_iq_{ij}\mathbb{E}(W_i)\mathbb{E}(W_j)\Big\}\\
			=&\frac{1}{2}\sum_{i,j}P_iq_{ij}[\mathbb{E}(W_i)-\mathbb{E}(W_j)]^2\ge 0.
		\end{split}
	\end{equation}
	Here the first equality is from Eq.~\ref{eqz2}, the third equality is from Eq.~\ref{eqz1}, and other equalities are equivalent transformations.
	
	Now we have  
	\begin{equation*}
		\begin{split}
			&\mathbb{E}(W^2)-\mathbb{E}(W)-[\mathbb{E}(W)]^2\\
			=&\sum_iP_i\mathbb{E}(W_i^2)-\sum_iP_i\mathbb{E}(W_i)-\Big[\sum_iP_i\mathbb{E}(W_i)\Big]^2\\
			=&\frac{1}{G}\sum_i F_iP_i\mathbb{E}(W_i)+\sum_i P_i\mathbb{E}(W_i)-\sum_iP_i\mathbb{E}(W_i)-\Big[\sum_iP_i\mathbb{E}(W_i)\Big]^2\\
			\ge &\sum_iP_i[\mathbb{E}(W_i)]^2-\Big[\sum_iP_i\mathbb{E}(W_i)\Big]^2\\
			=&\Big(\sum_iP_i\Big)\sum_iP_i[\mathbb{E}(W_i)]^2-\Big[\sum_iP_i\mathbb{E}(W_i)\Big]^2\ge 0,\\
		\end{split}
	\end{equation*}
	where the first equality is by definition, the second equality is from Eqs.~\ref{eq3},\ref{eq5}, the first inequality is from Eq.~\ref{eq6}, the third equality is from $\sum_i P_i=1$, and the second inequality is the Cauchy inequality.
	
	Since $\mathbb{E}(W^2)-[\mathbb{E}(W)]^2\ge \mathbb{E}(W)$, $\text{VMR}(W)=\{\mathbb{E}(W^2)-[\mathbb{E}(W)]^2\}/ \mathbb{E}(W)\ge 1$.
\end{proof}

We hypothesize that the requirement for $g_k(\cdot)$ in Proposition~\ref{np} can be dropped.

\begin{conjecture}
	\label{conj2}
	Assume $V_k$ is not contained in a directed cycle in the GRN, and $V_k$ does not regulate its interior factors. If $V_k$ has no autoregulation, meaning that $h_k(\cdot)$ does not depend on $n_k$ (but might depend on $\boldsymbol{n}_{\bar{k}}$), then $V_k$ has $\text{VMR}\ge 1$. 
\end{conjecture}

If the GRN has directed cycles, there is a conjecture by Paulsson et al. \cite{hilfinger2016constraints,yan2019kinetic}, which has been numerically verified but not proved yet. 
\begin{conjecture}
	Assume for each $V_i$, $g_i(\cdot)$ does not depend on $\boldsymbol{n}$, and $f_i(\cdot)$ does not depend on $n_i$ (no autoregulation). Then for at least one gene $V_j$, we have $\text{VMR}\ge 1$.
	\label{conj3}
\end{conjecture}
Notice that Conjecture~\ref{conj3} does not hold if $g_i$ depends on $\boldsymbol{n}_{\bar{i}}$. One counterexample is $m=2$, $f_1(n_2)=g_1(n_2)=1$ for $n_2=2$, $f_1(n_2)=g_1(n_2)=0$ for $n_2\ne 2$, and $f_2(n_1)=g_2(n_1)=1$ for $n_1=2$, $f_2(n_1)=g_2(n_1)=0$ for $n_1\ne 2$. Then $\text{VMR}=2e/(4e-1)\approx 0.55$ for both genes.

Assume Conjecture~\ref{conj3} is correct. For $m$ genes, if we find that VMR for each gene is less than $1$, then we can infer that autoregulation exists, although we do not know which gene has autoregulation.

\section{Applying theoretical results to experimental data}
\label{app}

\begin{algorithm}[!htbp]
	\caption{Detailed workflow of inferring autoregulation with gene expression data.}
	\label{alg}
	\vspace{-\bigskipamount}
	\ \\
	\begin{enumerate}
		\item \textbf{Input} 
		
		\quad Single-cell non-interventional one-time expression data for genes $V_1,\ldots,V_m$
		
		\quad The structure of GRN that contains $V_1,\ldots,V_m$
		\item \textbf{Calculate} the VMR of each $V_k$\
		
		\item \textbf{If} $V_k$ is not in a directed cycle (like PIP2 in Fig.~\ref{grn}) and $\text{VMR}<1$
		
		\quad   \textbf{Output} $V_k$ has autoregulation
		
		\quad   // Assume the degradation of $V_k$ is not regulated by $V_1,\ldots,V_m$
		
		\textbf{Else} 
		
		\quad   \textbf{If} $V_k$ has no ancestor in the GRN (like PIP3 in Fig.~\ref{grn}) and $\text{VMR}>1$
		
		\quad\quad   \textbf{Output} $V_k$ has autoregulation
		
		\quad\quad //Assume the expression of $V_k$ is autonomous
		
		\quad   \textbf{Else}
		
		\quad\quad   \textbf{Output} We cannot determine whether $V_k$ has autoregulation
		
		\quad   \textbf{End} of if
		
		\textbf{End} of if

	\end{enumerate}
\end{algorithm}

We summarize our theoretical results into Algorithm~\ref{alg}. Proposition~\ref{prop2} applies to a gene that has no ancestor in the GRN. However, it requires the corresponding gene has autonomous expression, which is difficult to validate and often does not hold in reality. Thus the inference result by Proposition~\ref{prop2} for $\text{VMR}>1$ (positive autoregulation) is not very reliable. When $\text{VMR}<1$ and Proposition~\ref{prop2} could apply, we should instead apply Proposition~\ref{np} to determine the existence of autoregulation, since Proposition~\ref{np} does not require the expression to be autonomous, thus being much more reliable. Proposition~\ref{np} applies when the gene is not in a feedback loop and has $\text{VMR}<1$. Notice that our result cannot determine that a gene has no autoregulation. 

For a given gene without autoregulation, its expression level satisfies a Poisson distribution, and VMR is $1$. If we have $n$ samples of its expression level, then the sample VMR (sample variance divided by sample mean) asymptotically satisfies a Gamma distribution $\Gamma[(n-1)/2,2/(n-1)]$, and we can determine the confidence interval of sample VMR \cite{eden2010drawing}. If the sample VMR is out of this confidence interval, then we know that VMR is significantly different from $1$, and Propositions~\ref{prop2},\ref{np} might apply. 

We apply our method to four groups of single-cell non-interventional one-time gene expression data from experiments, where the corresponding GRNs are known. Notice that we need to convert indirect measurements into protein/mRNA count. See Table~\ref{tab} for our inference results and theoretical/experimental evidence that partially validates our results. See Appendix~\ref{app1} for details. There are 186 genes in these four data sets, and we can only determine that 12 genes have autoregulation (7 genes determined by Proposition~\ref{prop2}, and 5 genes determined by Proposition~\ref{np}). Not every VMR is less than $1$, so that Conjecture~\ref{conj3} does not apply. For the other 174 genes, Proposition~\ref{prop2} and Proposition~\ref{np} do not apply, and we do not know whether they have autoregulation. 

In some cases, we have experimental evidence that some genes have autoregulation, so that we can partially validate our inference results. Nevertheless, as discussed in the Introduction, there is no gold standard to evaluate our inference results. 

In the data set by Guo et al. \cite{guo2010resolution}, Sanchez-Castillo et al. \cite{sanchez2018bayesian} inferred that 17 of 39 genes have autoregulation, and 22 genes do not have autoregulation. We infer that 5 genes have autoregulation, and 34 genes cannot be determined. Here 3 genes are shared by both inference results to have autoregulation. Consider a random classifier that randomly picks 5 genes and claims they have autoregulation. Using Sanchez-Castillo et al. as the standard, this random classifier has probability $62.55\%$ to be worse than our result, and $10.17\%$ to be better than our result. Thus our inference result is better than a random classifier, but the advantage is not significant.

\begin{table}[]	
	\begin{tabular}{lllll}
		Source          &\begin{tabular}[c]{@{}l@{}}Propo-\\ sition~\ref{prop2}\end{tabular}                                                     & \begin{tabular}[c]{@{}l@{}}Propo-\\ sition~\ref{np}\end{tabular}                                                              & Theory                                                                 & Experiment                                                     \\
		\hline
		\begin{tabular}[c]{@{}l@{}}Guo\\ et al. \cite{guo2010resolution} \end{tabular}
		& \begin{tabular}[c]{@{}l@{}}FN1\\ \textbf{HNF4A}\end{tabular}               & \begin{tabular}[c]{@{}l@{}} \textbf{TCFAP2C}\\ \textbf{BMP4}\\ CREB312\end{tabular} & \begin{tabular}[c]{@{}l@{}}BMP4 \cite{sanchez2018bayesian}\\ HNF4A \cite{sanchez2018bayesian}\\  TCFAP2C \cite{sanchez2018bayesian}\end{tabular} & \begin{tabular}[c]{@{}l@{}}BMP4 \cite{pramono2016thrombopoietin}\\ HNF4A \cite{chahar2014chromatin}\\ TCFAP2C \cite{kidder2010examination}\end{tabular} \\
		\hline
		\begin{tabular}[c]{@{}l@{}}Psaila\\ et al. \cite{psaila2016single}\end{tabular}
		& \begin{tabular}[c]{@{}l@{}}BIM\\ CCND1\\ \textbf{ECT2}\\ PFKP\end{tabular} &                                                                     &                                                                        & ECT2 \cite{hara2006cytokinesis}                                                           \\
		\hline
		\begin{tabular}[c]{@{}l@{}}Moignard\\ et al. \cite{moignard2015decoding}\end{tabular}
		&                                                                   & \begin{tabular}[c]{@{}l@{}}EIF2B1\\ HOXD8\end{tabular}          &                                                                        &                                                                \\
		\hline
		\begin{tabular}[c]{@{}l@{}}Sachs\\ et al. \cite{sachs2005causal}\end{tabular}
		& PIP3                                                              &                                                                          &                                                                        &                                                               
	\end{tabular}
	\caption{The autoregulation inference results by our method on four data sets. Source column is the paper that contains this data set. Proposition~\ref{prop2} column is the genes that can be only inferred by Proposition~\ref{prop2} to have autoregulation. Proposition~\ref{np} column is the genes that can be inferred by Proposition~\ref{np} to have autoregulation. Theory column is the genes inferred by both our method and other theoretical works to have autoregulation. Experiment column is the genes inferred by both our method and other experimental works to have autoregulation. \textbf{Bold} font means the inferred gene with autoregulation is validated by other results. Details can be found in Appendix~\ref{app1}.}
	\label{tab}
\end{table}



\section{Conclusions}
\label{con}
For a single gene that is not affected by other genes, or a group of genes that form a connected GRN, we develop theoretical results to determine the existence of autoregulation. These results generalize known relationships between autoregulation and VMR by dropping restrictions on parameters. Our results only depend on VMR, which is easy to compute and more robust than other complicated statistics. We also apply our method to experimental data and detect some genes that might have autoregulation. We prove two propositions for Markov chains, which might have theoretical values. 

We introduce two conjectures that have been numerically verified but not yet proved. They are of theoretical interest and worth further consideration. The Markov chain models in this paper can be studied via lifting into a higher-dimensional space \cite{wang2020mathematical}, treating as a random dynamical system \cite{ye2016stochastic}, or as a branching process \cite{jiang2017phenotypic}. With the expression profiles of different genes, we can construct a similarity graph \cite{wang2021inference}. If we know the existence of autoregulation for some genes, we can use the similarity graph to infer other genes.

Our method requires \textbf{independent} and \textbf{identically} distributed samples from the \textbf{exact} \textbf{stationary} distribution of a \textbf{fully observed} Markov chain, plus a known \textbf{GRN}. Proposition~\ref{prop2} requires the expression is \textbf{autonomous}. Proposition~\ref{np} requires that the GRN has \textbf{no} directed \textbf{cycle}, and \textbf{degradation} is \textbf{not regulated}. If our inference fails, then some requirements are not met: (1) cells might affect each other, making the samples dependent; (2) cells are heterogeneous; (3) the measurements have extra errors; (4) the cells are not at stationary; (5) there exist unobserved variables that affect gene expression; (6) the GRN is inferred by a theoretical method, which can be interfered by the existence of autoregulation; (7) the expression is non-autonomous; (8) the GRN has unknown directed cycles; (9) the degradation rate is regulated by other genes. Such situations, especially the unobserved variables, are unavoidable. Therefore, current data might not satisfy these requirements, and our inference results should be interpreted as informative findings, not ground truths. In fact, other theoretical works that determine gene autoregulation, or general gene regulation, also need similar assumptions and might fail. Nevertheless, with the development of experimental technologies, there will be more data with higher quality that fit the requirements of our method. Thus we believe that our method will be more applicable in the future.

Cells keep growing and dividing, and the gene expression fluctuates along the cell cycle. Discussions on such non-stationary situations can be found in other papers \cite{cao2020analytical,swain2002intrinsic,skinner2016single,jia2021frequency}.

About cell heterogeneity, we prove a result in Appendix~\ref{app2} that if several cell types have $\text{VMR}\ge 1$, then for a mixed population of such cell types, we still have $\text{VMR}\ge1$. Therefore, cell heterogeneity does not fail Proposition~\ref{np}, since $\text{VMR}<1$ for the mixture of several cell types means $\text{VMR}<1$ for at least one cell type.

\appendix
\section{Details of applications on experimental data}
\label{app1}
In experiments, the expression levels of genes are not directly measured as mRNA or protein counts. Rather, they are measured as cycle threshold (Ct) values or fluorescence intensity values. Such indirected measurements need to be converted. Related details can be found in other papers \cite{jia2017stochastic}.

Guo et al. \cite{guo2010resolution} measured the expression (mRNA) levels of 48 genes for mouse embryo cells at different developmental stages. We consider three groups (16-cell stage, 32-cell stage, 64-cell stage) that have more than 50 samples. Sanchez-Castillo et al. \cite{sanchez2018bayesian} used such data to infer the GRN structure, including autoregulation, but the GRN only contains 39 genes. Thus we ignore the other 9 genes. In the inferred GRN, genes BMP4, CREB312, and TCFAP2C are not contained in directed cycles. In the 16-cell stage group with 75 samples, if there is no autoregulation, then the $95\%$ confidence interval of VMR is $[0.7041,1.3470]$. BMP4 ($\text{VMR}=0.2139$), CREB312 ($\text{VMR}=0.1971$), and TCFAP2C ($\text{VMR}=0.3468$) have significantly small VMR, and we can apply Proposition~\ref{np} to infer that BMP4, CREB312, and TCFAP2C might have autoregulation. In the other two groups, these genes do not have $\text{VMR}<1$, and the results are relatively weak. Besides, in the inferred GRN, genes FN1 and HNF4A have no ancestors. For the 16-cell stage with 75 samples, the VMR of FN1 and HNF4A are $3.4522$ and $1.3599$, outside of the $95\%$ confidence interval $[0.7041,1.3470]$; for the 32-cell stage with 113 samples, the VMR of FN1 and HNF4A are $93.1070$ and $46.7688$, outside of the $95\%$ confidence interval $[0.7554,1.2784]$; for the 64-cell stage with 159 samples, the VMR of FN1 and HNF4A are $117.3059$ and $93.9589$, outside of the $95\%$ confidence interval $[0.7917,1.2322]$. Thus we can apply Proposition~\ref{prop2} to infer that FN1 and HNF4A ($\text{VMR}>1$ for all three cell groups) might have positive autoregulation. Nevertheless, it is more likely that the expressions of FN1 and HNF4A are non-autonomous, and there is no autoregulation. Sanchez-Castillo et al. \cite{sanchez2018bayesian} inferred that BMP4, HNF4A, TCFAP2C have autoregulation. Besides, there is experimental evidence that BMP4 \cite{pramono2016thrombopoietin}, HNF4A \cite{chahar2014chromatin}, TCFAP2C \cite{kidder2010examination} have autoregulation. Therefore, our inference results are partially validated.

Psaila et al. \cite{psaila2016single} measured the expression (mRNA) levels of 90 genes for human megakaryocyte-erythroid progenitor cells. Chan et al. \cite{chan2017gene} inferred the GRN structure (autoregulation not included). In the inferred GRN, genes BIM, CCND1, ECT2, PFKP have no ancestors. BIM has 214 effective samples, and VMR is $187.7$, outside of the $95\%$ confidence interval $[0.8191,1.1987]$. CCND1 has 68 effective samples, and VMR is $111.3$, outside of the $95\%$ confidence interval $[0.6905,1.3660]$. ECT2 has 56 effective samples, and VMR is $8.2$, outside of the $95\%$ confidence interval $[0.6618,1.4069]$. PFKP has 134 effective samples, and VMR is $82.1$, outside of the $95\%$ confidence interval $[0.7742,1.2543]$. Thus we can apply Proposition~\ref{prop2} to infer that BIM, CCND1, ECT2, PFKP might have positive autoregulation. Nevertheless, it is more likely that the expressions of these four genes are non-autonomous, and there is no autoregulation. There is experimental evidence that ECT2 has autoregulation \cite{hara2006cytokinesis}, which partially validates our inference results. No other gene fits the requirement of Proposition~\ref{np}.

Moignard et al. \cite{moignard2015decoding} measured the expression (mRNA) levels of 46 genes for mouse embryo cells. Chan et al. \cite{chan2017gene} inferred the GRN structure (autoregulation not included). Gene EIF2B1 has 3934 effective samples, and VMR is $0.66$, outside of the $95\%$ confidence interval $[0.9563,1.0447]$. Gene EIF2B1 has 12 effective samples, and VMR is $0.24$, outside of the $95\%$ confidence interval $[0.3469,1.9927]$. We can apply Proposition~\ref{np} to infer that EIF2B1 and HOXD8 might have autoregulation. No other gene fits the requirement of Proposition~\ref{prop2}.

Sachs et al. \cite{sachs2005causal} measured the expression (protein) levels of 11 genes in the RAF signaling pathway for human T cells. The measurements were repeated for 14 groups of cells under different interventions. Werhli et al. \cite{werhli2006comparative} inferred the GRN structure (autoregulation not included). In the inferred GRN (Fig.~\ref{grn}), PIP3 gene has no ancestor, and its VMRs in all 14 groups are larger than $5$, while the $95\%$ confidence intervals for all 14 groups are contained in $[0.8,1.2]$. Therefore, we can apply Proposition~\ref{prop2} and infer that PIP3 might have positive autoregulation. Nevertheless, it is more likely that the expression of PIP3 is non-autonomous, and there is no autoregulation. No other gene fits the requirement of Proposition~\ref{np}.

\section{Heterogeneity and VMR}
\label{app2}
\begin{proposition}
	Consider $n$ independent random variables $X_1,\ldots,X_n$ and probabilities $p_1,\ldots,p_n$ with $\sum p_i=1$. Consider an independent random variable $R$ that equals $i$ with probability $p_i$. Construct a random variable $Z$ that equals $X_i$ when $R=i$. If each $X_i$ has $\text{VMR}\ge 1$, then $Z$ has $\text{VMR}\ge 1$.
\end{proposition}
\begin{proof}
	We only need to prove this for $n=2$. The case for general $n$ can be proved by mixing two variables iteratively. 
	
	Consider random variables $X,Y$ and construct $Z$ that equals $X$ or $Y$ with probability $p$ or $1-p$. Since $\text{VMR}(X)\ge 1$, $\text{VMR}(Y)\ge 1$, we have $\mathbb{E}(X^2)-[\mathbb{E}(X)]^2\ge \mathbb{E}(X)$ and $\mathbb{E}(Y^2)-[\mathbb{E}(Y)]^2\ge \mathbb{E}(Y)$. Then 
	
	\begin{equation*}
		\begin{split}
			&\text{VMR}(Z)\\
			=&\frac{p\mathbb{E}(X^2)+(1-p)\mathbb{E}(Y^2)}{p\mathbb{E}(X)+(1-p)\mathbb{E}(Y)}\\
			&+\frac{-p^2[\mathbb{E}(X)]^2-2p(1-p)\mathbb{E}(X)\mathbb{E}(Y)-(1-p)^2[\mathbb{E}(Y)]^2}{p\mathbb{E}(X)+(1-p)\mathbb{E}(Y)}\\
			=&\frac{p\mathbb{E}(X^2)-p[\mathbb{E}(X)]^2+(1-p)\mathbb{E}(Y^2)-(1-p)[\mathbb{E}(Y)]^2}{p\mathbb{E}(X)+(1-p)\mathbb{E}(Y)}\\
			&+\frac{p(1-p)[\mathbb{E}(X)]^2-2p(1-p)\mathbb{E}(X)\mathbb{E}(Y)+p(1-p)[\mathbb{E}(Y)]^2}{p\mathbb{E}(X)+(1-p)\mathbb{E}(Y)}\\
			\ge & \frac{p\mathbb{E}(X)+(1-p)\mathbb{E}(Y)}{p\mathbb{E}(X)+(1-p)\mathbb{E}(Y)}+\frac{p(1-p)[\mathbb{E}(X)-\mathbb{E}(Y)]^2}{p\mathbb{E}(X)+(1-p)\mathbb{E}(Y)}\\
			\ge & 1.
		\end{split}
	\end{equation*}

\end{proof}

\section*{Acknowledgments}
This research was partially supported by NIH grant R01HL146552 (Y.W.). Y.W. would like to thank Jiawei Yan for fruitful discussions, and Xiangting Li, Zikun Wang, Mingtao Xia for helpful comments. The authors would like to thank some anonymous reviewers for their wise suggestions.

\section*{Declaration of interests}
The Authors declare that there is no conflict of interest.

\bibliographystyle{vancouver}
\bibliography{AC}

\end{document}